\pdfoutput=1

\documentclass[orivec]{llncs}

\usepackage{graphicx}
\usepackage{amsmath}
\usepackage{amssymb}
\usepackage{times}

\newcommand{\Thta}{\mathrm{\Theta}}
\newcommand{\Omga}{\mathrm{\Omega}}
\newcommand{\bigO}{\mathrm{O}}
\newcommand{\omga}{\mathrm{\omega}}
\newcommand{\RDelta}{\mathrm{\Delta}}
\newcommand{\ord}{ord}
\newcommand{\unord}{unord}

\begin{document}

\pagestyle{headings}  

\mainmatter              

\title{Heaps Simplified}

\author{Bernhard Haeupler\inst{2} \and Siddhartha Sen\inst{1}$^{,4}$ \and Robert E. Tarjan\inst{1}\inst{,3}$^{,4}$}

\institute{Princeton University, Princeton NJ 08544, \email{\{sssix, ret\}@cs.princeton.edu}
\and
CSAIL, Massachusetts Institute of Technology, \email{haeupler@mit.edu}
\and
HP Laboratories, Palo Alto CA 94304}

\maketitle              

\addtocounter{footnote}{+1}
\footnotetext{Research at Princeton University partially supported by
NSF grants CCF-0830676 and CCF-0832797 and US-Israel Binational
Science Foundation grant 2006204. The information contained herein
does not necessarily reflect the opinion or policy of the federal
government and no official endorsement should be inferred.}

\begin{abstract}

The heap is a basic data structure used in a wide variety of
applications, including shortest path and minimum spanning tree
algorithms.  In this paper we explore the design space of
comparison-based, amortized-efficient heap implementations.  From a
consideration of dynamic single-elimination tournaments, we obtain the
binomial queue, a classical heap implementation, in a simple and
natural way.  We give four equivalent ways of representing heaps
arising from tournaments, and we obtain two new variants of binomial
queues, a one-tree version and a one-pass version.  We extend the
one-pass version to support key decrease operations, obtaining the
{\em rank-pairing heap}, or {\em rp-heap}.  Rank-pairing heaps combine
the performance guarantees of Fibonacci heaps with simplicity
approaching that of pairing heaps.  Like pairing heaps, rank-pairing
heaps consist of trees of arbitrary structure, but these trees are
combined by rank, not by list position, and rank changes, but not
structural changes, cascade during key decrease operations.

\end{abstract}

\section{Introduction} \label{sec:intro}

A {\em meldable heap} (henceforth just a {\em heap}) is a data structure
consisting of a set of items, each with a real-valued key, that
supports the following operations:

\begin{itemize}

\item $make\mbox{-}heap$: return a new, empty heap.

\item $insert(x, H)$: insert item $x$, with predefined key, into heap $H$.

\item $find\mbox{-}min(H)$: return an item in heap $H$ of minimum key.

\item $delete\mbox{-}min(h)$: delete from heap $H$ an item of minimum key,
and return it; return null if the heap is empty.

\item $meld(H_1, H_2)$: return a heap containing all the items in
disjoint heaps $H_1$ and $H_2$, destroying $H_1$ and $H_2$.

\end{itemize}

\noindent Some applications of heaps need either or both of the following
additional operations:

\begin{itemize}

\item $decrease\mbox{-}key(x, \RDelta, H)$: decrease the key of item
$x$ in heap $H$ by amount $\RDelta > 0$, assuming $H$ is the unique
heap containing $x$.

\item $delete(x, H)$: delete item $x$ from heap $H$, assuming $H$ is
the unique heap containing $x$.

\end{itemize}

We shall assume that all keys are distinct; if they are not, we can
break ties using any total order of the items.  We allow only binary
comparisons of keys, and we study the amortized
efficiency~\cite{tarjan1985acc} of heap operations.  To obtain a bound
on amortized efficiency, we assign to each configuration of the data
structure a non-negative {\em potential}, initially zero.  We define
the {\em amortized time} of an operation to be its actual time plus
the change in potential it causes.  Then for any sequence of
operations the sum of the actual times is at most the sum of the
amortized times.

Since $n$ numbers can be sorted by doing $n$ insertions into an
initially empty heap followed by $n$ minimum deletions, the classical
$\Omga(n\log n)$ lower bound~\cite[p. 183]{Knuth1973b} on the number
of comparisons needed for sorting implies that either insertion or
minimum deletion must take $\Omga(\log n)$ amortized time, where $n$
is the number of items currently in the heap.  For simplicity in
stating bounds we assume $n \ge 2$. We investigate simple data
structures such that minimum deletion (or deletion of an arbitrary
item if this operation is supported) takes $\bigO(\log n)$ amortized
time, and each of the other supported heap operations takes $\bigO(1)$
amortized time.  These bounds match the lower bound.  (The logarithmic
lower bound can be beaten by using multiway
branching~\cite{Fredman1994,han2002isn}.)

Many heap implementations have been proposed over the years.  See
e.g. \cite{1328914}.  We mention only those directly related to our
work.  The {\em binomial queue} of Vuillemin~\cite{359478} supports
all the heap operations in $\bigO(\log n)$ worst-case time per
operation.  This structure performs quite well in
practice~\cite{Brown1978}.  Fredman and Tarjan~\cite{Fredman1987}
invented the {\em Fibonacci heap} specifically to support key decrease
operations in $\bigO(1)$ time, which allows efficient implementation
of Dijkstra's shortest path
algorithm~\cite{dijkstra1959ntp,Fredman1987}, Edmonds' minimum
branching algorithm~\cite{Edmonds1967,Gabow1986}, and certain minimum spanning
tree algorithms~\cite{Fredman1987,Gabow1986}.  Fibonacci heaps support
deletion of the minimum or of an arbitrary item in $\bigO(\log n)$
amortized time and the other heap operations in $\bigO(1)$ amortized
time.  They do not perform well in practice,
however~\cite{Liao1992,Moret1994}.  As a result, a variety of
alternatives to Fibonacci heaps have been
proposed~\cite{driscoll1988rha,violation_heaps,hyer1995gti,kaplan1999nhd,1328914,Peterson1987},
including a self-adjusting structure, the {\em pairing
heap}~\cite{fredman1986phn}.

Pairing heaps support all the heap operations in $\bigO(\log n)$
amortized time and were conjectured to support key decrease in
$\bigO(1)$ amortized time. Despite empirical evidence supporting the
conjecture~\cite{Jones1986,Liao1992,Stasko1987}, Fredman~\cite{320214}
showed that it is not true: pairing heaps and related data structures
that do not store subtree size information require $\Omga(\log\log n)$
amortized time per key decrease.  Whether pairing heaps meet this
bound is open; the best upper bound is $\bigO(2^{2\sqrt{\lg\lg
n}})$~\cite{Pettie_pairingheaps}\footnote{We denote by $\lg$ the
base-two logarithm.}.  Very recently ElMasry~\cite{1496822} proposed a
more-complicated alternative to pairing heaps that does have an
$\bigO(\log\log n)$ amortized bound per key decrease.

Fredman's result gives a time-space trade-off between the number of
bits per node used to store subtree size information and the amortized
time per key decrease; at least $\lg\lg n$ bits per node are needed to
obtain $\bigO(1)$ amortized time per key decrease.  Fibonacci heaps
use $\lg\lg n + 2$ bits per node; relaxed heaps~\cite{driscoll1988rha}
use only $\lg\lg n$ bits per node.  All the structures so far proposed
that achieve $\bigO(1)$ amortized time per key decrease do a cascade
of local restructuring operations to keep the underlying trees
balanced.

The bounds of Fibonacci heaps can be obtained in the worst case, but
only by making the data structure more complicated: run-relaxed
heaps~\cite{driscoll1988rha} and fat heaps~\cite{kaplan1999nhd}
achieve these bounds except for melding, which takes $\bigO(\log n)$
time worst-case; a very complicated structure of
Brodal~\cite{brodal1996wce} achieves these bounds for all the heap
operations.

Our goal is to systematically explore the design space of
amortized-efficient heaps and thereby discover the simplest possible
data structures.  As a warm-up, we begin in Section 2 by showing that
binomial queues can be obtained in a natural way by dynamizing
balanced single-elimination tournaments. We give four equivalent ways
of representing heaps arising from tournaments, and we give two new
variants of binomial queues: a one-tree version and a one-pass
version.

Our main result is in Section 3, where we extend our one-pass version
of binomial queues to support key decrease.  Our main insight is that
it is not necessary to maintain balanced trees; all that is needed is
to keep track of tree sizes, which we do by means of ranks.  We call
the resulting data structure a {\em rank-pairing heap}, or {\em
rp-heap}.  The rp-heap achieves the bounds of Fibonacci heaps with
simplicity approaching that of pairing heaps. It resembles the lazy
variant of pairing heaps~\cite[p. 125]{fredman1986phn}, except that
trees are combined by rank, not by list position.  In an rp-heap, rank
changes can cascade but not structural changes, and the trees in the
heap can evolve to have arbitrary structure.  We study two types of
rp-heaps.  One is a little simpler and uses $\lg\lg n$ bits per node
to store ranks, exactly matching Fredman's lower bound, but its
analysis is complicated and yields larger constant factors; the other
uses $\lg\lg n + 1$ bits per node and has small constant factors.

In Section \ref{sec:simpler-kd} we address the question of whether key
decrease can be further simplified. We close in
Section~\ref{sec:remarks} with open problems.

\section{Tournaments as Heaps} \label{sec:tourn}

Given a set of items with real-valued keys, one can determine the item
of minimum key by running a single-elimination tournament on the
items.  To run such a tournament, repeatedly match pairs of items
until only one item remains.  To match two items, compare their keys,
declare the item of smaller key the winner, and eliminate the
loser. Break a tie in keys using a total order of the items.

Such a tournament provides a way to represent a heap that supports
simple and efficient implementations of all the heap operations.  The
winner of the tournament is the item of minimum key.  To insert a new
item into a non-empty tournament, match it against the winner of the
tournament.  To meld two non-empty tournaments, match their winners.
To delete the minimum from a tournament, delete its winner and run a
tournament among the items that lost to the winner.

For the moment we measure efficiency by counting comparisons.  Later
we shall show that our bounds on comparisons also hold for running
time, to within additive and multiplicative constants.  Making a heap
or finding the minimum takes no comparisons, an insertion or a meld
takes one.  The only expensive operation is minimum deletion.  To make
this operation efficient, we restrict matches by using {\em ranks}.
Each item begins with a rank of zero.  After a match between items of
equal rank, the winner's rank increases by one.  Such a match is {\em
fair}.  A match between items of unequal rank is {\em unfair}.  After
an unfair match the winner's rank does not change.  Optionally, if the
winner has rank less than that of the loser, its rank can increase to
any value no greater than the loser's rank.  One rule governs a
tournament: match two items of equal rank if possible. We call such a
tournament \emph{balanced}.  Except for implementation details, this
is the entire description of the data structure.

\begin{lemma}
\label{lem:b-heap-size}
A balanced tournament whose winner has rank $k$ contains at least $2^k$ items.
\end{lemma} 

\begin{proof}

The proof is by induction on the number of matches.  The lemma holds
initially.  An item's rank increases to $k$ only when the item has
rank $k - 1$ and it beats an item of rank $k - 1$, or it beats an item
of rank at least $k$. In the former case the two matched items are the
winners of disjoint balanced tournaments, each containing at least
$2^{k-1}$ items.  In the latter case the tournament won by the loser
contains at least $2^k$ items.  After the match the combined
tournament contains at least $2^k$ items. \qed

\end{proof}

In stating bounds we denote by $n$ the number of items in a tournament
or heap; we assume $n > 1$.  To bound the amortized number of
comparisons per heap operation, we define the potential of a balanced
tournament to be the number of unfair matches.  An insertion or meld
increases the potential by at most one and hence takes at most two
amortized comparisons.  The tournament rule and
Lemma~\ref{lem:b-heap-size} guarantee that at most $\lg n$ unfair
matches occur during a minimum deletion, since an unfair match can
only occur when there is at most one item per rank, and by
Lemma~\ref{lem:b-heap-size} the number of ranks is at most $\lg n +
1$.  (The minimum rank is 0; the maximum is $\lfloor \lg n \rfloor$.) 
To analyze a minimum deletion, let $x$ be the deleted winner.
Consider the sequence of matches won by $x$.  Each such match was
either unfair or increased the rank of $x$.  Thus if there were $k$
such unfair matches, the total number of matches won by $x$ was at
least $k$ and at most $k + \lg n$.  Deletion of $x$ eliminates these
matches.  Finding a new winner takes at most $k + \lg n$ matches, of
which at most $\lg n$ are unfair.  Hence the amortized number of
comparisons done by the minimum deletion is at most $k + \lg n - k +
\lg n = 2\lg n$.

It remains to implement the data structure.  The standard way to
represent a tournament is by a full binary tree, whose leaves contain
the items and whose internal nodes represent the matches.  Each
internal node contains the winner of the match and its rank after the
match.  The internal nodes containing an item form a path of the
matches it won.  The tree is {\em heap ordered}: the item in a node
has minimum key among the items in the descendants of the node.  We
call this the {\em full representation}. (See Figure \ref{fig:reps}a.)

\begin{figure}[h]
\centering
\begin{tabular}{ll}
\includegraphics[scale=0.22]{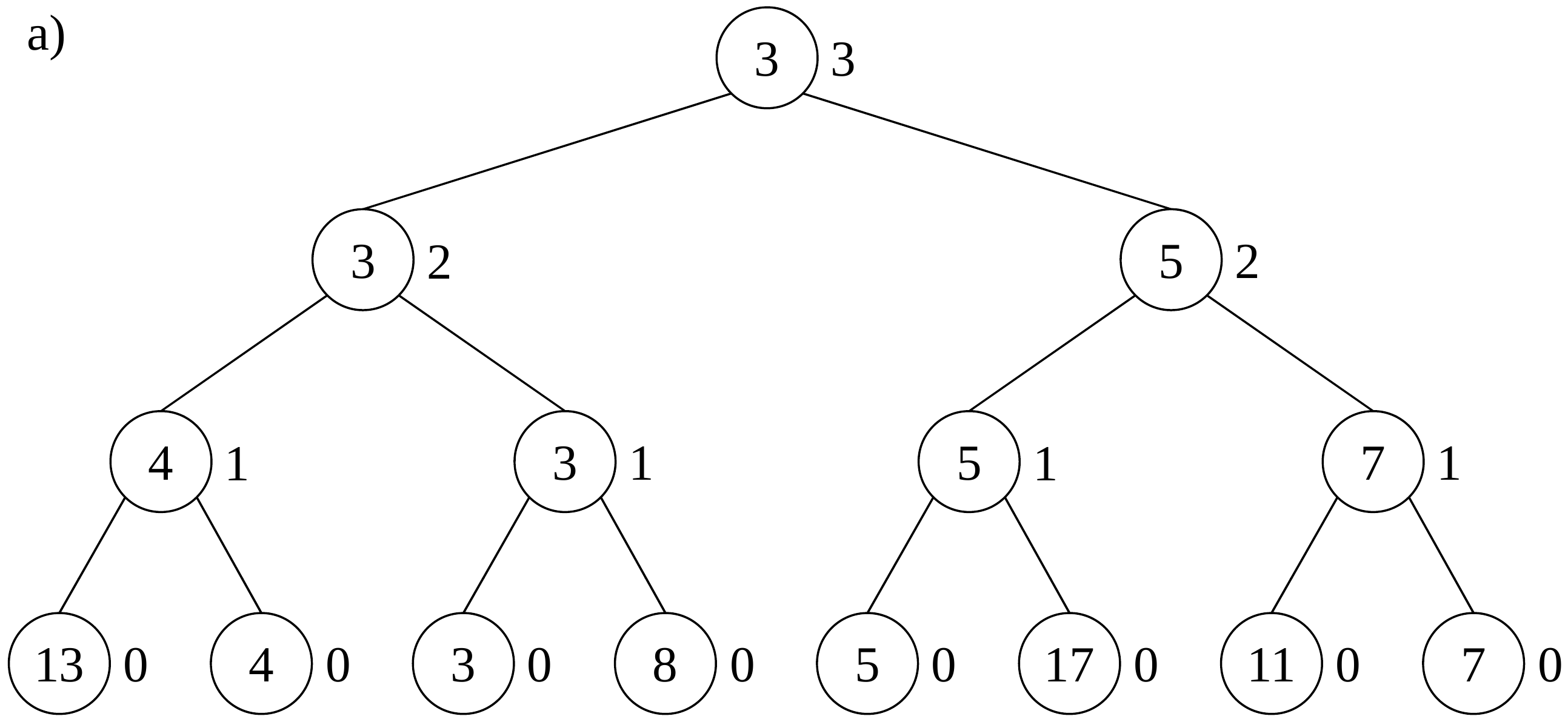} & \hspace{0.1in}
\includegraphics[scale=0.22]{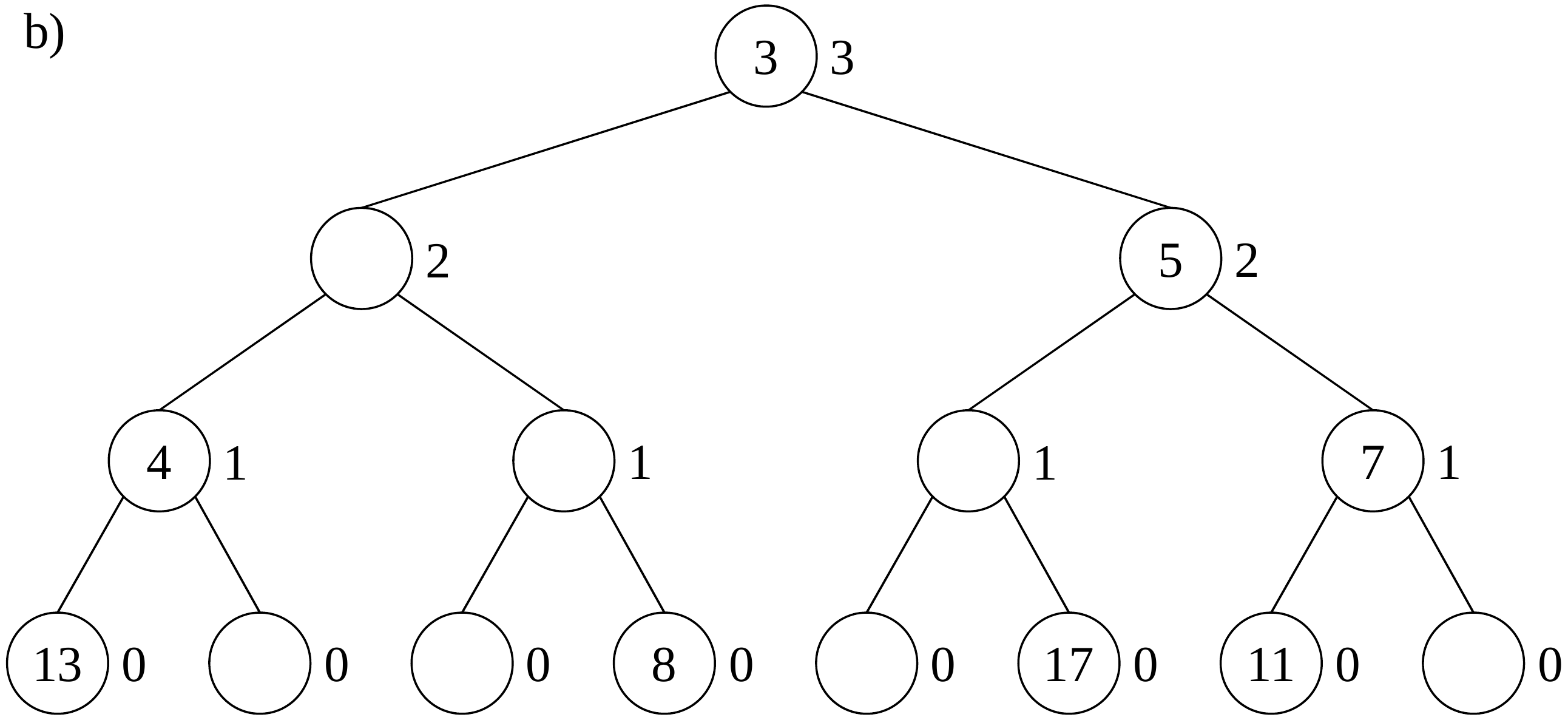} \\
\\
\includegraphics[scale=0.22]{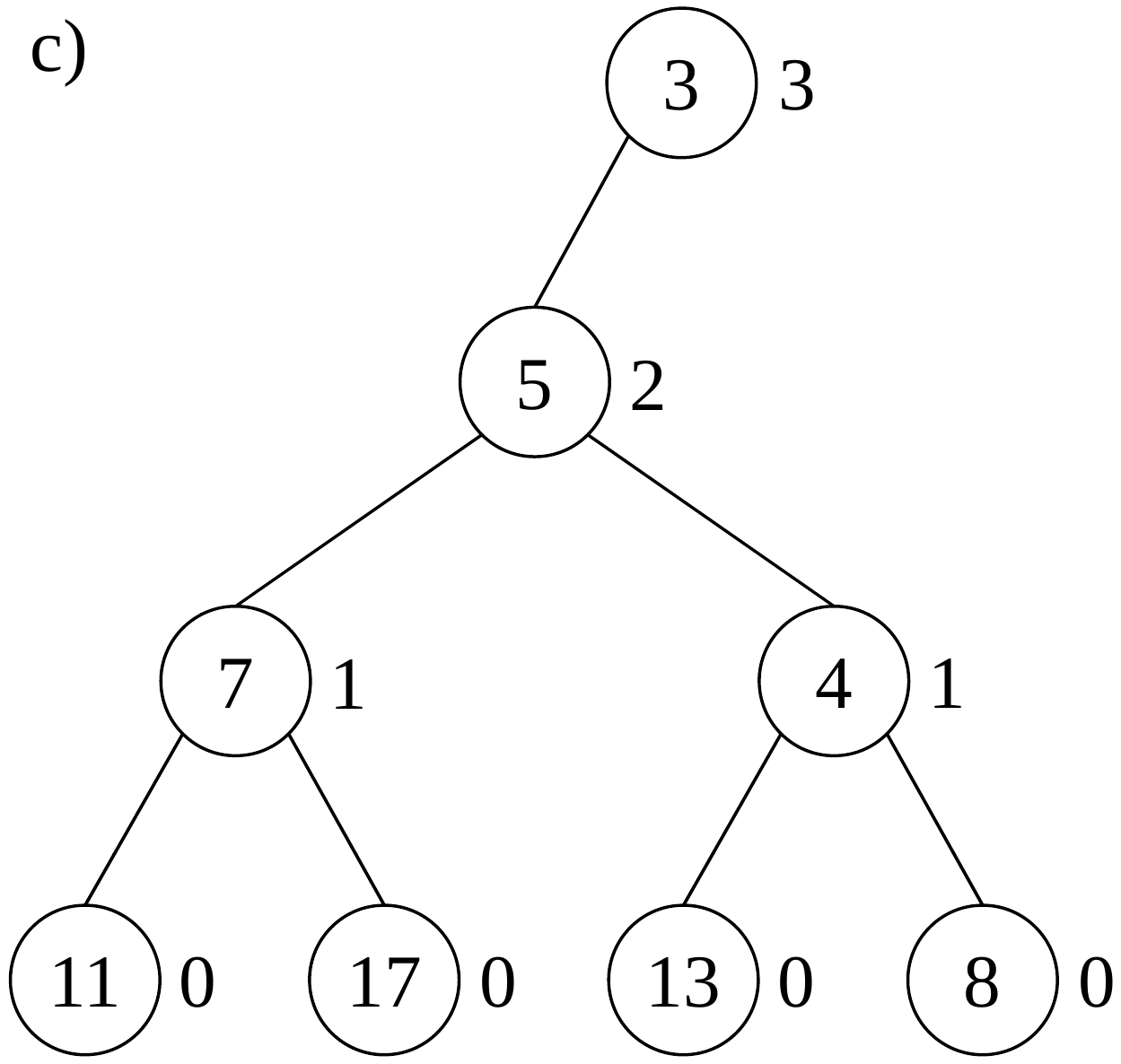} & \hspace{0.1in}
\includegraphics[scale=0.22]{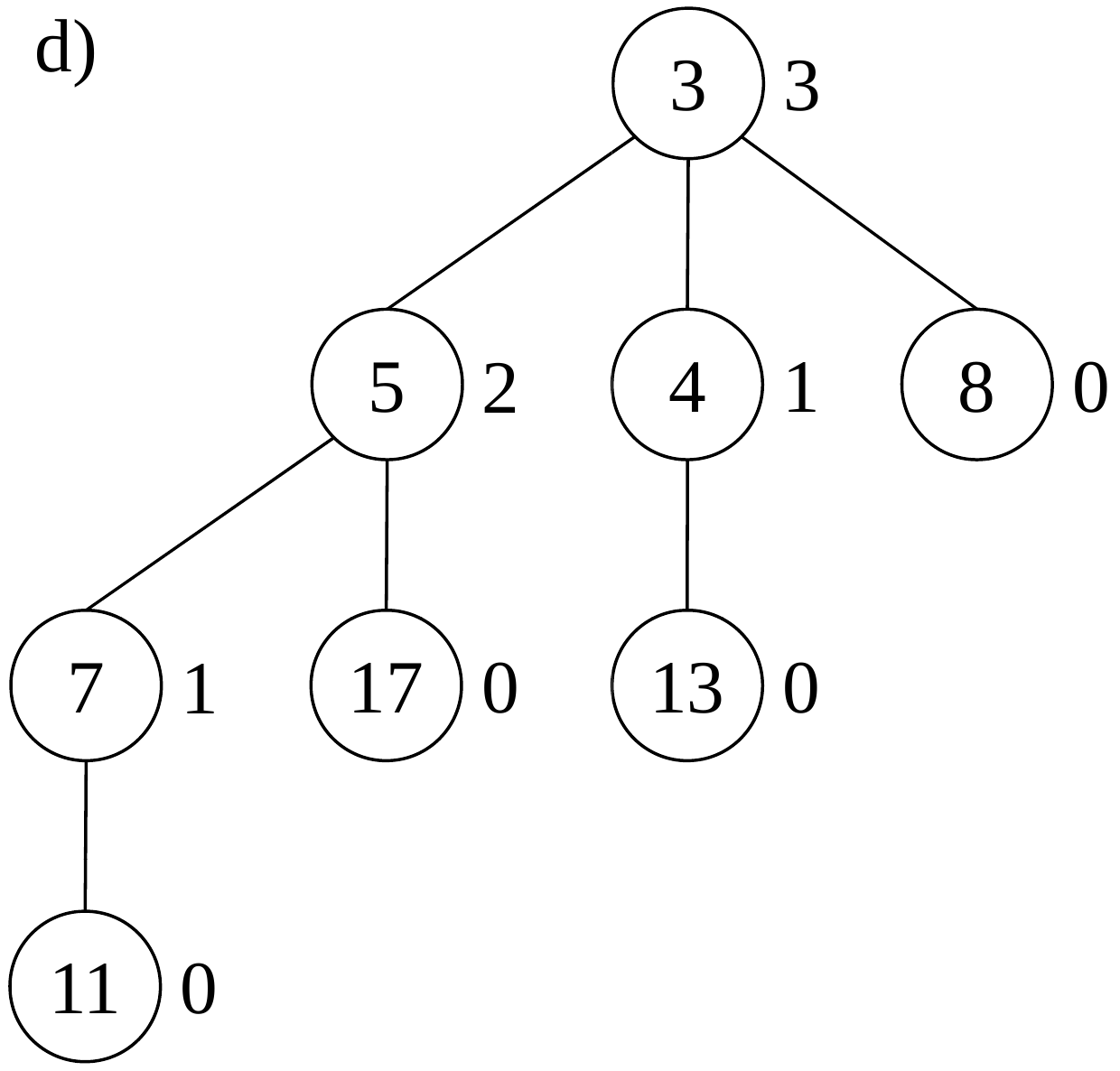}
\end{tabular}
\caption{Four representations of a tournament: a) full, b) half-empty,
c) half-ordered, and d) heap-ordered. The number to the right of each
node is its rank.}
\label{fig:reps}
\end{figure}

This representation uses almost twice as many nodes as necessary.  As
a first step in reducing the space, we remove each item from all but
the highest node containing it.  Our tree is now {\em half empty}: the
root is full, and each parent has two children, one full, the other
empty.  Each node, whether full or empty, has a rank.  This is the
{\em half-empty representation}.  (See Figure \ref{fig:reps}b.)

A half-empty heap-ordered tree represents a tournament transparently,
but it uses as many nodes as the full representation.  We obtain a
more compact representation by eliminating the empty nodes.  To do
this we define the {\em ordered} child and the {\em unordered} child
of a full node to be its full child and the full child of its empty
sibling, respectively.  The tree becomes a {\em half tree}: the root
has an ordered child but no unordered child, every non-root node has
both an ordered child and an unordered child, and any child can be
missing.  We give each child a rank equal to the rank of its parent in
the half-empty representation, minus one.  (If the parent represents a
fair match, this is just the rank of the child itself in the
half-empty representation, but it need not be if the parent represents
an unfair match.  We define the ranks in this way to avoid losing
information about the ranks in the half-empty representation if there
are unfair matches.)  With this definition a leaf has rank zero.  We
adopt the convention that a missing child has rank $-1$ and a root has
rank one larger than that of its child.  (There is no need to
explicitly maintain the ranks of roots.)  The rank of a half tree is
the rank of its root.  The {\em ordered} and {\em unordered subtrees}
of a parent are the subtrees rooted at its ordered and unordered
children, respectively.  Each subtree is {\em half ordered}: the item
in a node has smaller key than those of all items in its ordered
subtree.  The items that lost to an item in a node are those in the
nodes on the path starting from the ordered child of the node and
descending through unordered children.  This is the {\em half-ordered
representation} of a heap. (See Figure \ref{fig:reps}c.)
 
In the half-ordered representation we do not need to move items among
nodes as matches take place, so the items can {\em be} the nodes: the
data structure can be {\em endogenous}~\cite{Tarjan1983}.  Henceforth
we shall assume an endogenous representation.

We obtain our fourth and final representation, the {\em heap-ordered
representation}, by viewing a half tree as the binary tree
representation of a tree~\cite[pp. 332-346]{Knuth1973}: the ordered
and unordered children of a node become its first child and next
sibling, respectively.  The tree is heap-ordered; the children of a
node are the items it defeated, most recent match first.  (See Figure
\ref{fig:reps}d.)

Many heap structures, including binomial queues~\cite{359478},
Fibonacci heaps~\cite{Fredman1987}, pairing
heaps~\cite{fredman1986phn}, and relaxed heaps~\cite{driscoll1988rha},
were originally presented in the heap-ordered representation.  In
their paper on pairing heaps, Fredman et al.~\cite{fredman1986phn}
described the half-ordered representation and observed that it is the
binary tree representation~\cite{Knuth1973} of a heap-ordered tree.
Later, Dutton~\cite{dutton:whd} used the half-ordered representation
in his {\em weak-heap} data structure, and H{\o}yer~\cite{hyer1995gti}
proposed various kinds of half-ordered balanced trees as heaps.
The version of Fredman et al.~\cite{fredman1986phn} matches Knuth's
definition~\cite{Knuth1973}: the left and right children of a node are
its first child and next sibling in the heap-ordered representation,
respectively.  The version of Dutton and Hoyer reverses left and
right.  To avoid confusion we have named the children based on their
roles.  Throughout we denote by $\ord(x)$ and $\unord(x)$ the ordered
and unordered child of $x$, respectively.
 
Among the four representations we prefer the half-ordered
representation because it saves space and simplifies the
implementation of key decrease operations.  To implement balanced
tournaments in the half-ordered representation, we store with each
node its rank and pointers to its ordered and unordered children.  To
match the roots $x$ and $y$ of two half trees, compare their keys,
make the ordered subtree of the winner the unordered subtree of the
loser, and make the new half tree rooted at the loser the ordered
subtree of the winner.  Set the rank of the winner's new child
appropriately.  (This is the rank of the winner in the full
representation, minus one.)  (See Figure \ref{fig:match}.)

\begin{figure}[h]
\centering
\includegraphics[scale=0.25]{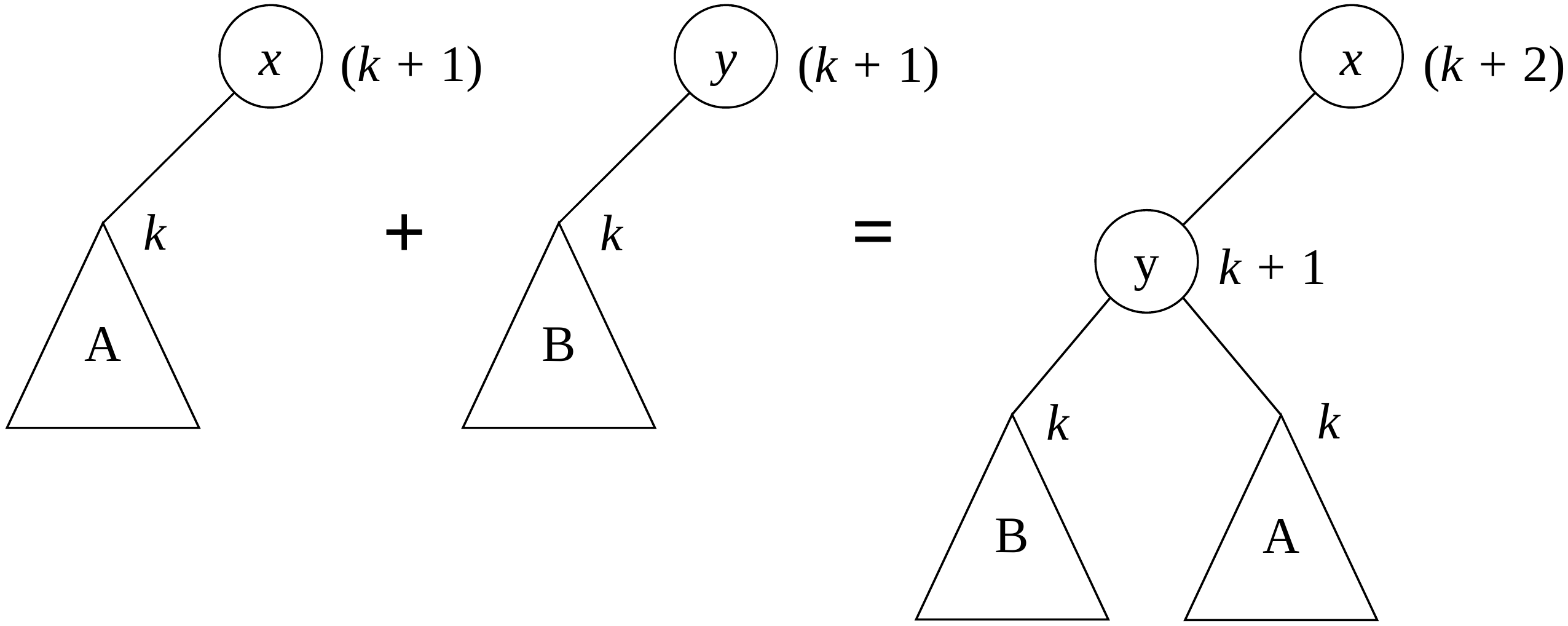}
\caption{A fair match between roots $x$ and $y$ of two half-trees.
Ranks are to the right of nodes.  Ranks of roots (in parentheses) need
not be maintained explicitly.}
\label{fig:match}
\end{figure}

A match of two roots takes $\bigO(1)$ time.  Thus insertion and melding
take $\bigO(1)$ time.  To do a minimum deletion, form the new heap by
starting at the child of the root and walking down through unordered
children.  Detach each node along with its ordered subtree, which
together form a half tree.  For each rank, keep track of a half tree
of that rank.  When a new half tree has the same rank as an existing
one, match their roots, producing a single half tree of one higher
rank.  If there is already a half tree of one higher rank, match its
root and the root of the new half tree.  Continue until the remaining half
tree is the only one of its rank.  Repeat this process for each new
half tree.  Once there is at most one half tree per rank, repeatedly
match any two of the remaining roots until only one root remains.
This method rebuilds the half tree after a minimum deletion in $\bigO(k + \log
n)$ time, where $k$ is the number of unfair matches won by the deleted
winner.

The amortized analysis of the number of comparisons per heap operation
extends to the running time and gives amortized time bounds of $\bigO(1)$
for insertion and melding and $\bigO(\log n)$ for minimum deletion.

A balanced tournament is a one-tree version of a binomial queue.  In
binomial queues, only fair matches are done, and a heap is represented
by a set of half trees rather than by a single half tree.  Fair
matches guarantee that each half tree is {\em perfect}: the ordered
subtree of the root is a perfect binary tree.  During an insertion,
meld, or minimum deletion, fair matches are done until there is at
most one half tree per rank.  Each heap operation takes $\bigO(\log
n)$ time in the worst case.  By allowing unfair matches, we are able
to represent a heap by a single tree.  The cost of this simplicity is
a linear worst-case time bound for minimum deletion.

We obtain another version of binomial queues by avoiding unfair
matches, doing fair matches only during minimum deletions, and running
only a single pass of matches during a minimum deletion.  The
resulting data structure, which we call the {\em one-pass binomial
queue}, has the same amortized efficiency as the one-tree version.
Here are the details.

A one-pass binomial queue consists of a set of heap-ordered perfect
half trees and a pointer to the root of minimum key, the {\em minimum
node}.  We represent the set of half trees by a singly-linked circular
list of the roots, with the minimum node first.  To insert an item
into a heap, create a new, one-node half tree containing the item, add
this half tree to the existing set of half trees, and update the
minimum node.  To meld two heaps, unite their sets of half trees and
update the minimum node.  To delete the minimum in a heap, take apart
the half tree rooted at the minimum node by deleting it and walking
down the path from its child through unordered children, making each
node on the path together with its ordered subtree into a half tree.
Group the half trees into a maximum number of pairs of equal rank and
run the corresponding matches.  (Each half tree except at most one per
rank participates in one match.)  Update the minimum node.

To analyze one-pass binomial queues we define the potential of a heap
to be the number of half trees it contains.  Lemma
\ref{lem:b-heap-size} holds for one-pass binomial queues, so no rank
exceeds $\lg n$.  A make-heap, find-min, or meld operation takes
$\bigO(1)$ time and does not change the potential.  An insertion takes
$\bigO(1)$ time and increases the potential by one.  Thus each of
these operations takes $\bigO(1)$ amortized time.  Consider a minimum
deletion.  Disassembling the half tree rooted at the minimum node
increases the number of trees and the potential by at most $\lg n$.
Let $h$ be the number of half trees after the disassembly but before
the round of fair matches.  The total time for the minimum deletion is
$\bigO(h)$, including the time to pair the trees by rank.  There are
at least $(h - \lg n)/2 - 1$ fair matches, reducing the potential by
at least this amount.  If we scale the running time so that it is at
most $h/2$, then the amortized time of the minimum deletion is
$\bigO(\log n)$.  (Scaling the running time is equivalent to
multiplying the potential by a constant factor.)

The same analysis applies if we do arbitrary additional fair matches
during a minimum deletion.  The extreme case is to continue doing fair
matches until at most one half tree per rank remains, as in the
original version of binomial queues.

\section{Key Decrease and Arbitrary Deletion} \label{sec:key-dec}

Our next goal is to add key decrease as an $\bigO(1)$-time operation.
Once key decrease is supported, one can delete an arbitrary item by
decreasing its key to $-\infty$ and doing a minimum deletion.  

A parameter of both key decrease and arbitrary deletion is the heap
containing the given item.  If the application does not provide this
information and melds occur, one needs a separate data structure to
maintain the partition of items into heaps.  With such a data
structure, the time to find the heap containing a given item is small
but not $\bigO(1)$~\cite{Kaplan2002}.

Fibonacci heaps~\cite{Fredman1987} were invented specifically to
support key decrease efficiently, but they require two extra pointers
per node (in the heap-ordered representation, to the parent and
previous sibling) and they are slower in practice than other heap
implementations~\cite{Liao1992,Moret1994}.  All the alternatives to
Fibonacci heaps with the same amortized efficiency,
including~\cite{driscoll1988rha,violation_heaps,hyer1995gti,1328914,Peterson1987},
have the property that decreasing a key can cause a cascade of tree
restructuring, possibly delayed.  But such restructuring is
unnecessary, as we show by introducing a new data structure, the
rank-pairing heap, or rp-heap, in which the trees have arbitrary
structure and the only cascading is of rank changes.

To obtain rp-heaps, we modify one-pass binomial queues to support key
decrease.  To the half-ordered representation we add parent pointers.
We also relax the requirement that the rank of a child be exactly one
less than that of its parent.  Let $p(x)$ and $r(x)$ be the parent and
rank of node $x$, respectively.  The {\em rank difference} of $x$ is
$r(p(x)) - r(x)$; the rank difference is defined only if $x$ has a
parent.  We maintain the ranks so that each leaf has rank zero and the
two children of a non-root have rank differences $1$ and $1$, or $0$
and at least $1$.  (One of these children can be missing; the rank of
a missing node is $-1$.)  We call a half tree with ranks that obey
this rule a {\em type-1 half tree}, and we call a set of heap-ordered
type-1 half trees a {\em type-1 rank-pairing heap} or {\em rp-heap}.

Lemma \ref{lem:b-heap-size} extends to type-1 half trees, so the
maximum rank of a node in a type-1 rp-heap is $\lg n$, and the number
of bits needed to store ranks is $\lg\lg n$ per node, matching
Fredman's lower bound.

We represent an rp-heap by a singly-linked circular list of the roots
of its half trees, with the minimum node first.  We implement making a
heap, insertion, melding, and minimum deletion exactly as in one-pass
binomial queues.  Fair matches maintain the rank rule. (See Figure
\ref{fig:match}.)

We implement key decrease as follows.  (See Figure
\ref{fig:decrease-key}.)  To decrease the key of item $x$ in rp-heap
$H$ by $\RDelta$, subtract $\RDelta$ from the key of $x$ and update the
minimum node.  If $x$ is a root, stop.  Otherwise, let $y$ be the
unordered child of $x$.  Detach the subtrees rooted at $x$ and $y$,
and reattach the subtree rooted at $y$ in place of the original
subtree rooted at $x$.  Add the half tree rooted at $x$ to the set of
trees representing $H$.  There may now be a violation of the rank rule
at $p(y)$, whose new child, $y$, may have lower rank than $x$, the
node it replaces.  To check for a violation and restore the rank rule
if necessary, let $u = p(y)$ and repeat the following step until it
stops: \\

\noindent {\em Decrease rank (type 1)}: If $u$ is the root, stop.
Otherwise, let $v$ and $w$ be the children of $u$.  Let $k$ equal
$r(v)$ if $r(v) > r(w)$, $r(w)$ if $r(w) > r(v)$, or $r(w) + 1$ if
$r(v) = r(w)$.  If $k = r(u)$, stop.  Otherwise, let $r(u) = k$ and $u
= p(u)$. \\

\begin{figure}[h]
\centering
\includegraphics[scale=0.4]{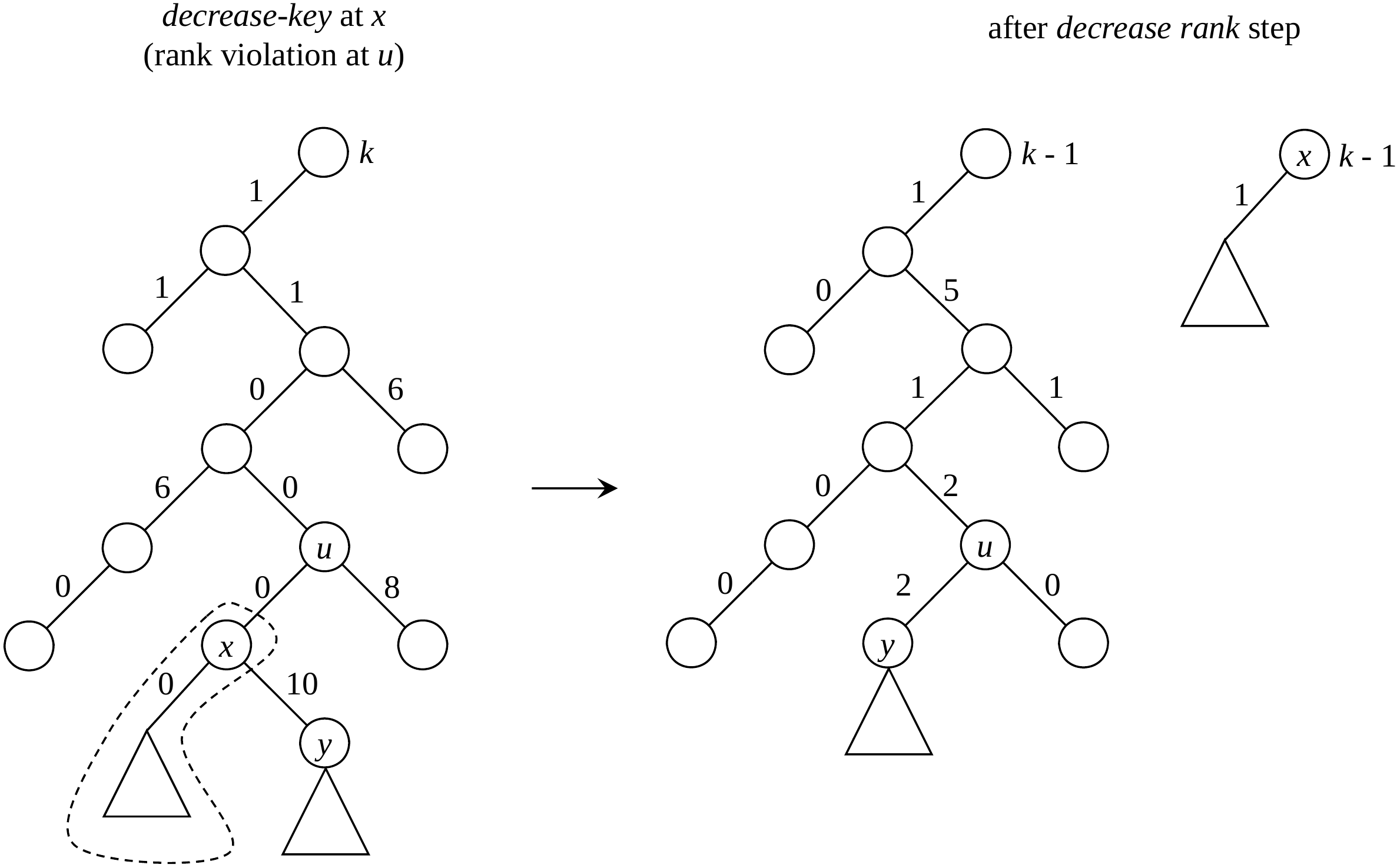}
\caption{Key decrease in a type-1 rp-heap.}
\label{fig:decrease-key}
\end{figure}

If $u$ breaks the rank rule, it obeys the rule after its rank is
decreased, but $p(u)$ may not.  Rank decreases propagate up along a
path through the tree until all nodes obey the rule. Each successive
rank decrease is by the same or a smaller amount.  The definition of
{\em decrease rank} assumes that ranks of roots are not maintained
explicitly.

If we mildly restrict the fair matches done during minimum deletions,
we can prove that type-1 rp-heaps match the efficiency of Fibonacci
heaps, but the analysis is complicated and yields larger constant
factors than one would like.  Before doing this, we consider a variant
obtained by relaxing the rank rule.  Its analysis is simpler, with
small constant factors, and we expect it to be efficient in practice.

The new rank rule is that each leaf has rank zero, and the two
children of a non-root have rank differences 1 and 1, or 1 and 2, or 0
and at least 2.  We call a half tree with ranks that obey this rule a
{\em type-2 half tree}, and we call a set of heap-ordered type-2 half
trees a {\em type-2 rank-pairing heap}.  The implementations of the
heap operations are exactly the same on type-2 rp-heaps as on type-1
rp-heaps, except for key decrease, which is the same except that the
rank decrease step becomes the following: \\

\noindent {\em Decrease rank (type 2)}: If $u$ is the root, stop.  Otherwise,
let $v$ and $w$ be the children of $u$.  Let $k$ equal $r(v)$ if $r(v)
> r(w) + 1$, $r(w)$ if $r(w) > r(v) + 1$, or $\max\{r(v) + 1, r(w) +
1\}$ otherwise.  If $k = r(u)$, stop.  Otherwise, let $r(u) = k$ and
$u = p(u)$. \\

\begin{lemma} \label{lem:u-heap-size}
A type-2 half tree of rank $k$ contains at least $\phi^k$ items,
where $\phi = (1 + \sqrt{5})/2$ is the golden ratio.
\end{lemma}

\begin{proof}
A type-2 half tree whose root is not a leaf can be decomposed into two
type-2 half trees by detaching the ordered subtree of the root,
detaching the unordered subtree of the old child of the root, and
making the latter tree into the ordered subtree of the root.  It
follows from the rank rule that the minimum number of items $n_k$ in a
type-2 half tree of rank $k$ satisfies the recurrence $n_0 = 1, n_1 =
2, n_k = n_{k-1} + n_{k-2}$ for $k \ge 2$. This is the recurrence for
the Fibonacci numbers $F_k$, offset by two. Hence $n_k = F_{k+2} \ge
\phi^k$~\cite[p. 18]{Knuth1973}. \end{proof}

To analyze type-2 rp-heaps, we call a node {\em bad} if it has an
ordered child of rank difference two or more and {\em good} otherwise.
A node of rank one or more with a missing ordered child is bad.  Roots
are good.  Bad nodes can cause extra matches during minimum deletions;
they are created only by key decrease operations.

To analyze the heap operations, we define the potential of a node of
rank $k$ to be $k$, $k + 1$, or $k + 2$ if it is a good child, a bad
child, or a root, respectively.  We define the potential of a heap to
be the sum of the potentials of its nodes.

A make-heap, find-min, or meld operation takes $\bigO(1)$ time and does
not change the potential.  An insertion takes $\bigO(1)$ time and creates
one new root of rank zero, for a potential increase of 2.  Hence each
of these operations takes $\bigO(1)$ amortized time.

Consider the deletion of a minimum node $x$ of rank $k$.  Deletion of
$x$ reduces the potential by $k + 2$.  Consider the effect on the
potential of disassembling the half tree rooted at $x$.  Let $y$ be a
node on the path from $\ord(x)$ descending through unordered children.
There are three cases.  If $\unord(y)$ has rank difference zero, then
$y$ is bad, and making it into a root cannot increase its potential.
(Its rank decreases by at least $1$.)  If $\unord(y)$ has rank
difference $1$, $\ord(y)$ has rank difference at least one.  In this
case making $y$ into a root can increase its potential by at most $2$.
If $\unord(y)$ has rank difference $2$ or more, making $y$ into a root
can increase its potential by at most $3$.  In either of the last two
cases, $\unord(y)$ can be missing. It follows that the new roots
increase in potential by a total of at most $2k$, and the net
potential increase caused by the disassembly is at most $k - 2 \le
\log_{\phi} n - 2$ by Lemma \ref{lem:u-heap-size}.

Let $h$ be the number of half trees in the heap after the disassembly.
The entire minimum deletion takes $\bigO(h)$ time.  A fair match of
two half trees of rank $j$ results in a new half tree of rank $j + 1$
whose root has a child of rank $j$, resulting in a net potential drop
of one.  There are at least $(h - \log_{\phi} n)/2 - 1$ fair matches,
reducing the potential by at least this much.  The decrease in
potential caused by the entire minimum deletion is thus at least $h/2
- (3/2)\log_{\phi} n$.  If we scale the time for the minimum
deletion so that it is at most $h/2 + \bigO(\log n)$, then its
amortized time is $\bigO(\log n)$.
 
It remains to analyze key decrease.  Consider a key decrease at a node
$x$ of rank $k$.  If $x$ is a root, the key decrease takes $\bigO(1)$
time and has no effect on the potential.  Suppose $x$ is not a root.
Breaking the half tree containing $x$ into two half trees takes
$\bigO(1)$ time and increases the potential by at most 4: node $x$
begins with potential at least $k$ but once it becomes a root it has
potential at most $k + 3$, and the old parent of $x$ may become bad,
either because its ordered child is now missing or because its new
ordered child has smaller rank than its old ordered child.  
We scale the time so that each iteration of {\em decrease rank} takes
time at most 1.  Each iteration either decreases the rank of a node or
is the last iteration, so the amortized time of an iteration other
than the last is at most zero, not counting increases in potential
caused by newly bad nodes.  A node $u$ other than the old parent of
node $x$ can only become bad if $r(\ord(u))$ decreases by some amount
and then $r(u)$ decreases by less.  If $r(\ord(u))$ decreases by at
least 2, the corresponding potential decrease pays for both the rank
decrease step at $\ord(u)$ and for $u$ becoming bad.  If $r(\ord(u))$
decreases by 1, the rank decrease step at $u$ is the last. It follows
that key decrease takes $\bigO(1)$ amortized time.

The maximum rank for a given heap size is the same for type-2 rp-heaps
as for Fibonacci heaps, namely at most $\log_{\phi} n$.  The
worst-case time for a key decrease in an rp-heap (of either type) is
$\Thta(n)$, as it is in a Fibonacci heap~\cite{Fredman1987}.  One can
reduce the worst-case time for a key decrease to $\bigO(1)$ by
delaying each such operation until the next minimum deletion.  This
requires keeping a list of possible minimum nodes that includes all
the roots and all the nodes whose keys have decreased since the last
minimum deletion. Making this change complicates the data structure
and is likely to worsen its performance in practice.

Now we turn to the analysis of type-1 rp-heaps.  We use the following
terminology.  A node is an $i$-{\em child} if it has rank difference
$i$, and an $i, j$-{\em node} if its children have rank differences
$i$ and $j$; this definition does not distinguish between ordered and
unordered children.  A leaf is a 1,1-node.  We call a node {\em green}
if it is a leaf but not a root, or it and both of its children are
1,1-nodes; {\em yellow} if it is a 0,1 node whose 0-child is a 1,
1-node, or it is a root with no child or with a 1, 1-child; and {\em
red} otherwise.  Red nodes can cause extra matches during minimum
deletions; they are created only by key decrease operations.  Yellow
nodes block the propagation of rank decreases: a rank decrease step on
a yellow node is terminating.  A key decrease can turn green nodes
into yellow nodes but not red ones, and can turn at most one yellow
node into a red node.

Our analysis requires a restriction on the fair matches.  We consider
two such restrictions; the first has a simpler analysis and the second
has a simpler, natural implementation.  First we assume that the fair
matches done during a minimum deletion preferentially match as many
red roots as possible.  One way to guarantee this is to maintain a
bucket for each rank, initially empty.  Process the half trees
one-at-a-time, beginning with the ones with red roots and finishing
with the others.  To process a half tree, insert it into the bucket
for its rank if this bucket is empty; if not, do a fair match of its
root with the root of the half tree in the bucket, and add the
resulting half tree to the output set, leaving the bucket empty.  Once
all the half trees have been processed, add to the output set any tree
remaining in a bucket.  This process assures that there is at most one
match per rank between a red root and a yellow root.

We define the potential of a node of rank $k$ to be $k$, $k + 2$, or
$k + 4$ if it is a green or yellow child, a yellow root, or it is red,
respectively.  The potential of a heap is the sum of the potentials of
its nodes.

The analysis of make-heap, find-min, meld, and insert is exactly the
same as for type-2 rp-heaps: each such operation takes $\bigO(1)$
amortized time; only insertions increase the potential, by two units
per insertion.

Consider the deletion of a minimum node $x$ of rank $k$.  Deletion of
$x$ reduces the potential by at least $k + 2$.  Disassembly of the
half tree rooted at $x$ produces a set of half trees, one per node on
the path from $\ord(x)$ descending through unordered children.  The
nodes on this path become roots; they are the only nodes whose
potentials can change.  Except for at most 2 units of potential, we
shall charge each such increase to a rank between 0 and $k - 1$, at
most 4 units per rank.  The total charge is then at most $4k + 2$, so
the potential increase caused by the disassembly is at most $4k + 2 -
(k + 2) = 3k$.

Let $y$ be a node on the path.  We consider five cases, one of which
has two subcases.  If $y$ is red and $\ord(y)$ is not a 0-child, then
its potential does not increase when it becomes a root; indeed, if $y$
becomes yellow, its potential drops by at least 2.  If $y$ is red and
$\ord(y)$ is a 0-child, then the potential of $y$ increases by 1 if it
becomes a red root and drops by 1 if it becomes a yellow root: its
rank increases by 1.  If $y$ is yellow and $\ord(y)$ is a 0-child, the
potential of $y$ increases by 3 when it becomes a root: its rank
increases by 1 and it stays yellow.  If $y$ is green, its potential
increases by 2 when it becomes a root: it becomes yellow.  In each of
these cases we charge any increase in potential to the rank of $y$.

If $y$ is yellow and $\unord(y)$ is a 0-child, the potential of $y$
increases by 4 if it becomes a red root or by 2 if it becomes a yellow
root.  From $y$ walk down the path of unordered children until coming
to a node $z$ that is not green or is the last node on the path.  Node
$z$ is a 1,1-node, so it is not yellow.  If it is red, its potential
does not increase when it becomes a root, since its rank does not
change.  We charge the potential increase of $y$ to the rank of $z$.
If $z$ is green, it is a leaf.  We charge 2 units of the potential
increase of $y$ to rank zero; any remainder (2 units) is uncharged.
The total charge to rank zero is 4, 2 for $y$ and 2 for $z$.  As
claimed, each rank is charged at most 4 units and at most 4 units are
uncharged.  We conclude that the potential increase caused by the
disassembly is at most $3k \le 3\lg n$.

Let $h$ be the number of half trees in the heap after the disassembly.
The entire minimum deletion takes $\bigO(h)$ time.  There are at least
$(h - \lg n)/2 - 1$ fair matches.  Each match of two red roots reduces
the potential by 1 since the winner becomes yellow; so does each match
of two yellow roots.  A match between a yellow root and a red root can
increase the potential by at most 1, but there are at most $\lg n$
such matches.  Thus the minimum deletion reduces the potential by at
least $(h - \lg n)/2 - 5\lg n - 1 = h/2 - (11/2)\lg n - 1$.  If we
scale the time for the minimum deletion so that it is at most $h/2$,
then the amortized time is $\bigO(\log n)$.

Finally, consider decreasing the key of node $x$.  The only nodes
whose potential can increase are $x$ and its ancestors.  If
$\unord(x)$ is a 0-child, there is only one rank decrease step, and
$x$ becoming a root increases its potential by at most 4.  If
$\unord(x)$ is not a 0-child, $x$ becoming a root increases its
potential by at most 3: its rank may increase by 1, but it cannot
become red.  (It may already be red.)  Let $u$ be a proper ancestor of
$x$ that becomes red as a result of the key decrease.  Node $u$ cannot
be green before the key decrease, it must be yellow.  If it is yellow,
its rank cannot decrease unless it is a root.  Thus it is the last
node subject to a key decrease step; its potential increases by at
most 4.  Each key decrease step except the last decreases the rank of
a node and hence the potential by at least 1.  If we scale the actual
time of a rank decrease step so it is at most 1, then the amortized
time of such a step is at most zero unless it is the last step.  Thus
key decrease takes $\bigO(1)$ amortized time.

A more natural way to do matches during a minimum deletion is to
preferentially pair half trees produced by the disassembly.  To do
this, use the same bucketing scheme as in the method that
preferentially pairs red roots, but process the half trees produced by
the disassembly first, in the order they are produced, followed by the
remaining half trees.  This method avoids the need to check root
colors, at an extra cost of at most 2 units of potential per key
decrease.

In order to analyze the method, we need a more elaborate potential
function.  We call a root {\em fresh} if it has just been created by a
disassembly, its rank is at least 1, and it has not yet participated
in a match.  A fresh root becomes {\em stale} once it wins a match or
the minimum deletion is finished.  Thus all roots are stale except in
the middle of minimum deletions.  We define the potential of a node of
rank $k$ to be $k$, $k + 2$, $k + 4$, or $k + 6$ if it is a green or
yellow child, a stale yellow root, a red child or a fresh root, or a
stale red root, respectively.  This gives fresh yellow roots and stale
red roots two extra units of potential.

These definitions have the following effect.  The winner of a fair
match is a stale yellow root.  Such a match reduces the potential by
at least 1 unless it matches a fresh red root against a stale yellow
root, in which case it increases the potential by at most 1.  After
the matches between fresh roots, each remaining fresh red root is
responsible for up to 2 units of extra potential, either because it is
matched against a stale yellow root or it becomes stale without being
matched.  There is at most one such root per positive rank.  We charge
the extra 2 units against the rank.

The analysis of make-heap, find-min, meld, and insert does not change.
The analysis of key decrease is the same except that the potential can
increase by an additional 2 units because a stale red root has 2 extra
units of potential; at most one node needs 2 extra units, the node
whose key decreases or the root of the original tree containing it,
but not both.

Consider the deletion of a minimum node $x$ of rank $k$.  In the
disassembly process we need to account for 2 extra units of potential
for each fresh yellow root.  This extra potential does not materially
affect the analysis unless the new root y was previously yellow and
$\ord(y)$ was a 0-child.  In this case y must be the last fresh node
of its rank, which means that there is no fresh red root of this rank
remaining after the matches between fresh roots.  We charge the extra
2 units against the rank.  The analysis of the matching process
proceeds as before; the constant factors are exactly the same.  The
extra 2 units of potential per rank in the new analysis correspond to
the at most $\lg n$ matches between red and yellow roots in the old
analysis.  We conclude that minimum deletion takes $\bigO(\log n)$
amortized time.

\section{Simpler Key Decreases?} \label{sec:simpler-kd}

It is natural to ask whether there is an even simpler way to decrease
keys while retaining the amortized efficiency of Fibonacci heaps.  We
give two answers: ``no'' and ``maybe''. We answer ``no'' by showing
that two possible methods fail.  The first method allows arbitrarily
negative but bounded positive rank differences.  With such a rank
rule, the rank decrease process following a key decrease need examine
only ancestors of the node whose key decreases, not their siblings.
Such a method can take $\Omga(\log n)$ time per key decrease, however,
as the following counterexample shows.  Let $b$ be the maximum allowed
rank difference.  Choose $k$ arbitrarily.  By means of a suitable
sequence of insertions and minimum deletions, build a heap that
contains a perfect half tree of each rank from 0 through $bk + 1$.
Let $x$ be the root of the half tree of rank $bk + 1$.  Consider the
path of unordered children descending from $\ord(x)$.  Decrease the
key of each node on this path whose rank is not divisible by $b$.
Each such key decrease takes $\bigO(1)$ time and does not violate the
rank rule, so no ranks change.  Now the path consists of $k + 1$
nodes, each with rank difference $b$ except the topmost.  Decrease the
keys of these nodes, smallest rank to largest.  Each such key decrease
will cause a cascade of rank decreases all the way to the topmost node
on the path.  The total time for these $k + 1$ key decreases is
$\Omga(k^2)$.  After all the key decreases the heap contains three
perfect half trees of rank zero and two of each rank from 1 through
$bk$.  A minimum deletion (of one of the roots of rank zero) followed
by an insertion makes the heap again into a set of perfect half trees,
one of each rank from 0 through $bk + 1$.  Each execution of this
cycle does $\bigO(\log n)$ key decreases, one minimum deletion, and
one insertion, and takes $\Omga(\log^2 n)$ time.

The second, even simpler method spends only $\bigO(1)$ time worst-case
on each key decrease, thus avoiding arbitrary cascading.  In this case
by doing enough operations one can build a half tree of each possible
rank, up to a rank that is $\omga(\log n)$.  Once this is done,
repeatedly doing an insertion followed by a minimum deletion (of the
just-inserted item) will result in each minimum deletion taking
$\omga(\log n)$ time.  Here are the details.  Suppose each key
decrease changes the ranks of nodes at most $d$ pointers away from the
node whose key decreases, where $d$ is fixed.  Choose $k$ arbitrarily.
By means of a suitable sequence of insertions and minimum deletions,
build a heap that contains a perfect half tree of each rank from 0
through $k$.  On each node of distance $d + 2$ or greater from the
root, in decreasing order by distance, do a key decrease with $\RDelta
= \infty$ followed by a minimum deletion.  No roots can be affected by
any of these operations, so the heap still consists of one half tree
of each rank, but each heap contains only $2^{d + 1}$ nodes, so there
are $\lfloor n/2^{d + 1} \rfloor$ half trees.  Now repeat the cycle of
an insertion followed by a minimum deletion.  Each such cycle takes
$\Omga(n/2^{d + 1})$ time.  The choice of ``$d + 2$'' in this
construction guarantees that no key decrease can reach the child of a
root, which implicitly stores the rank of the root.

This construction works even if we add extra pointers to the half
trees, as in Fibonacci heaps.  In a half tree, the {\em ordered
ancestor} of a node $x$ is the parent of the nearest ancestor of $x$
(including $x$ itself) that is an ordered child.  The ordered ancestor
corresponds to the parent in the heap-ordered representation.  Suppose
we augment half trees with ordered ancestor pointers.  Even for such
an augmented structure, the latter construction gives a bad example,
except that the size of a constructed half tree of rank $k$ is
$\bigO(k^{d + 1})$ instead of $\bigO(2^{d + 1})$, and each cycle of an
insertion followed by a minimum deletion takes $\Omga(n^{1/(d + 2)})$
time.

One limitation of this construction is that building the initial set
of half trees takes a number of operations exponential in the size of
the heap on which repeated insertions and minimum deletions are done.
Thus it is not a counterexample to the following question: is there a
fixed $d$ such that if each key decrease is followed by at most $d$
rank decrease steps (say of type 1), then the amortized time is
$\bigO(1)$ per insert, meld, and key decrease, and $\bigO(\log m)$ per
deletion, where $m$ is the total number of insertions?  A related
question is whether Fibonacci heaps without cascading cuts have these
bounds.  We conjecture that the answer is yes for some positive $d$,
perhaps even $d = 1$.  The following counterexample shows that the
answer is no for $d = 0$; that is, for the method in which a key
decrease changes no ranks except for the implicit ranks of roots.  For
arbitrary $k$, build a half tree of each rank from 0 through $k$, each
consisting of a root and a path of ordered children, inductively as
follows.  Given such half trees of ranks 0 through $k - 1$, insert an
item less than all those in the heap and then do $k$ cycles, each
consisting of an insertion followed by a minimum deletion that deletes
the just-inserted item.  The result will be one half tree of rank $k$
consisting of the root, a path of ordered children descending from the
root, a path $P$ of unordered children descending from the ordered
child of the root, and a path of ordered children descending from each
node of $P$; every child has rank difference 1.  (See Figure
\ref{fig:counter-ex1}.)  Do a rank decrease on each node of P.  Now
there is a collection of half trees of rank 0 through $k$ except for
$k - 1$, each a path.  Repeat this process on the set of half trees up
to rank $k - 2$, resulting in a set of half trees of ranks 0 through
$k$ with $k - 2$ missing.  Continue in this way until only rank 0 is
missing, and then do a single insertion.  Now there is a half tree of
each rank, 0 through $k$.  The total number of heap operations
required to increase the maximum rank from $k - 1$ to $k$ is
$\bigO(k^2)$, so in $m$ heap operations one can build a set of half
trees of each possible rank up to a rank that is $\Omga(m^{1/3})$.
Each successive cycle of an insertion followed by a minimum deletion
takes $\Omga(m^{1/3})$ time.

\begin{figure}[h]
\centering
\includegraphics[scale=0.4]{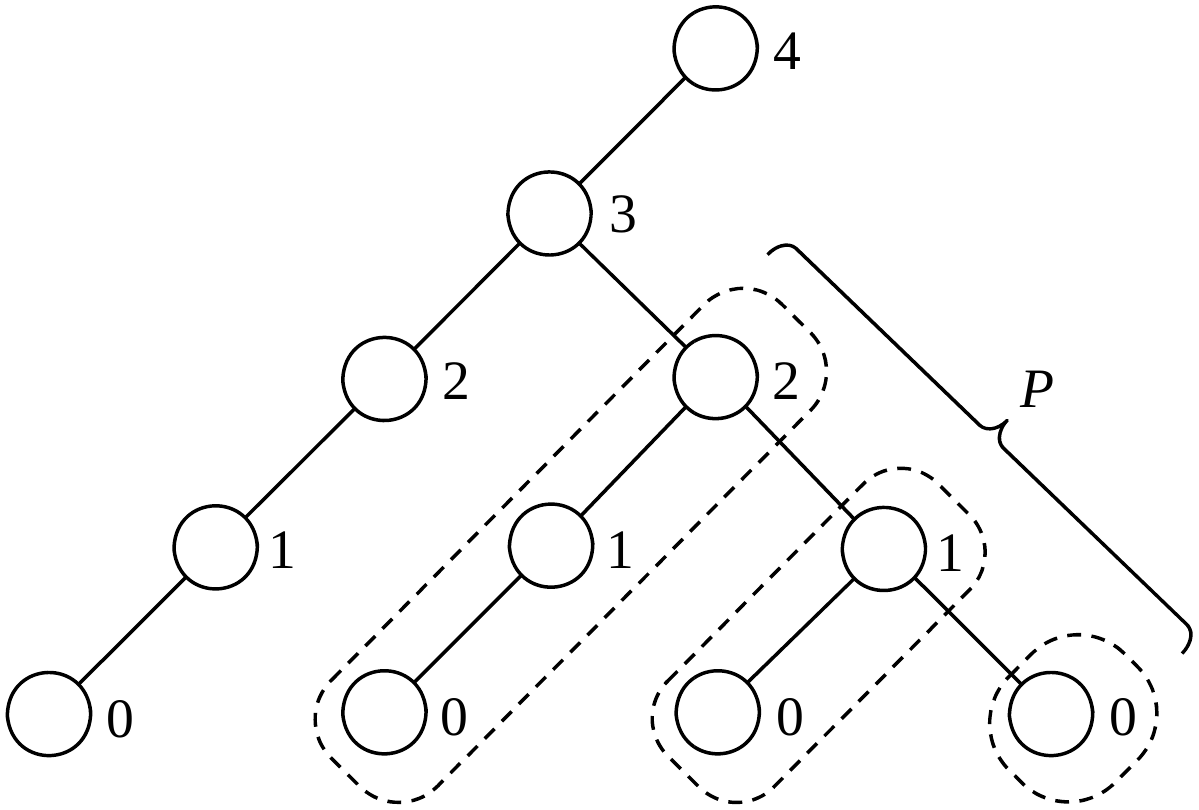}
\caption{A half tree of rank $k = 4$ buildable in $\bigO(k^3)$
operations if key decreases do not change ranks.  Key decreases on the
unordered children detach the circled subtrees.}
\label{fig:counter-ex1}
\end{figure}

\section{Remarks} \label{sec:remarks}

We have presented a new data structure, the rank-pairing heap, that
combines the performance guarantees of Fibonacci heaps with simplicity
approaching that of pairing heaps.  Our preliminary experiments
suggest that rank-pairing heaps may be competitive with pairing heaps
in practice, and we intend to do more-thorough experiments.  Several
theoretical questions remain: Can the restriction on fair matches in
type-1 rp-heaps be removed?  Can the constant factors in the analysis
of rp-heaps be reduced?  How is efficiency affected if only a constant
number of rank decrease steps are done after a key decrease?  Is there
a nice one-tree version of rp-heaps analogous to balanced tournaments?
Can Brodal's worst-case-efficient heap implementation be simplified?

\section*{Acknowledgement} \label{sec:ack}

We thank Haim Kaplan and Uri Zwick for extensive discussions that
helped to clarify the ideas in this paper, and for pointing out an
error in our original analysis of type-1 rp-heaps.

%
%
\bibliography{heaps}
\bibliographystyle{abbrv}

\end{document}